\documentclass [reqno, 12pt, a4paper] {amsart}

\usepackage{graphicx}
\usepackage{fancyhdr}
\usepackage{enumerate}
\usepackage{amsmath}
\usepackage{amsthm}
\usepackage{amssymb}
\usepackage{pdfsync}
\usepackage{hyperref}

\usepackage{changebar}
\usepackage{pdfsync}

\makeindex
\newtheorem{theorem}{Theorem}[section]

\newtheorem{proposition}[theorem]{Proposition}

\newtheorem{remark}[theorem]{Remark}

\newcommand{\mc}[1]{{\mathcal #1}}
\newcommand{\mf}[1]{{\mathfrak #1}}

\newcommand{\bs}[1]{{\boldsymbol #1}}

\newcommand\bp{{\mathbf p}}

\newcommand\ve{\varepsilon}

\newcommand\ft{\frak t}

\newcommand\br{{\mathbf r}}

\numberwithin{equation}{section}
\begin{document}

\title{Microscopic derivation of an isothermal thermodynamic transformation}
\author{Stefano Olla}
  \address{CEREMADE, UMR CNRS 7534\\
  Universit\'e Paris-Dauphine\\
  75775 Paris-Cedex 16, France \\
  \texttt{{ olla@ceremade.dauphine.fr}}}

\date{\today}
\begin{abstract}
We obtain macroscopic isothermal thermodynamic transformations by
space-time scalings of a microscopic Hamiltonian dynamics in contact
with a heat bath. The microscopic dynamics is given by a chain of anharmonic
oscillators subject to a varying tension (external force) and the
contact with the heat bath is modeled by independent Langevin dynamics
acting on each particle. After a diffusive space-time scaling and
cross-graining, the profile of volume converges to the solution of a
deterministic diffusive equation with boundary conditions given by the
applied tension. This defines an irreversible thermodynamic transformation from
an initial equilibrium to a new equilibrium given by the final tension
applied. Quasistatic reversible isothermal transformations are then
obtained by a further time scaling. Heat is defined as the total flux
of energy exchanged between the system and the heat bath. Then we
prove that the relation between the limit heat, work, free energy and
thermodynamic entropy agree with the first and second principle of
thermodynamics.  

\end{abstract}
 \thanks{This work has been partially supported by the
  European Advanced Grant {\em Macroscopic Laws and Dynamical Systems}
  (MALADY) (ERC AdG 246953). I thank Claudio Landim for stimulating
  conversations on quasi static limits and for the help in the proof of
  proposition \ref{quasistatic}}
\maketitle

\section{{Introduction}\label{sec:intro}}

Isothermal transformations are fondamental in thermodynamics, in
particular they are one of the components of the Carnot cycle. 
As often in thermodynamics, they represent \emph{idealized}
transformations where the system is maintained at a constant
temperature by being in constant contact with a \emph{large} heat reservoir
(\emph{heat bath)}. An isothermal thermodynamic transformations
connects two equilibrium states $A_0$ and $A_1$ at the same
temperature $T$, by changing the exterior forces applied. According to the
first law of thermodynamics, the change in the internal energy is
given by $U_1 - U_0 = W + Q$, where $W$ is the work made by the
exterior forces and $Q$ is the heat (energy) exchanged with the thermal
reservoir. The second law prescribes that the change of the free energy
$F= U - TS$ (where $S$ is the thermodynamic entropy), satisfy the Clausius
inequality $F_1 - F_0 \le W$, with equality satisfied for
\emph{reversible quasistatic} transformations. In the quasistatc
transformation we can then identify $Q = T(S_1 - S_0)$. 

The purpose of this article is to prove mathematically that the
termodynamic behavior of isothermal transformations, as described
above, can be obtained by proper space and time scaling of a
microscopic dynamics. We consider a one dimensional system, where the
equilibrium thermodynamic intensive parameters are given by the
temperature $T= \beta^{-1}$ and the tension
(or pressure) $\tau$, or by the extensive observables: length (volume)
$\mc L$ and energy $U$. This simplifies the problem as only two parameters are
needed to specify the equilibrium thermodynamic state and no phase
transitions will appear.

The microscopic model is given by a chain of $N$ anharmonic
oscillators, where the first particle is attached to a fix point and
on the last particle acts a force (tension) $\tilde\tau$, eventually
changing in time. The action of the thermal bath is modeled by
independent Langevin processes at temperature $T$, acting on each
particle. A mathematically equivalent model for the heat bath is given
by random collisions with the environment: at exponentially distributed
independent random times, each particle has a new velocity distributed
by a centered gaussian with variance $T$. 

As a consequence of the action of the thermal bath, the time
evolution of the microscopic configuration of the positions and
velocities of the particles is stochastic. The distance between the
first and the last particle defines the microscopic length of the
system, while the energy is given by the sum of the kinetic energies
of each particle and the potential energy of each spring. 

For each value of the applied tension $\tau$, the system has an
equilibrium probability distribution explicitely given by a Gibbs
measure, a product measure in this case. The temperature
parameter is fixed by the heat bath. Starting the system
with an equilibrium given by tension $\tau_0$, and changing the applied
tension to $\tau_1$, the system will go out of equilibrium before
reaching the new equilibrium state. During this transformation a
certain amount of energy is exchanged with the thermostats and
mechanical work is done by the force applied. We prove that, under a
proper macroscopic rescaling of space and time, all these (random)
quantities, converge to \textit{deterministic} values predicted by
thermodynamics. 

When the system is out of equilibrium, either for a change in the
tension applied, or by initial conditions, there is an evolution of the
local length (or stretch) on a diffusive macroscopic space-time
scale. This is governed by a diffusion equation that describe the
inhomogeneity of the system during the isothermal
transformation. After an infinite time (in this scale) it reach the
new equilibrium state given by a constant value of the local lenght,
corresponding to the value of the tension $\tau_1$. We have obtained,
in this diffusive time scale, an \emph{irreversible} thermodynamic 
transformation, that satisfies a strict Clausius inequality between work
and change of the free energy. Under a further rescaling of time, that
correspond in a slower change of the applied tension, we obtain a
\emph{reversible quasi-static} transformation that satisfies Clausius
inequality. In fact, for the irreversible transformation we obtain the
following relation between heat and changes of thermodynamic entropy $S$
\begin{equation*}
  Q = T \Delta S - \mathcal D
\end{equation*}
where $\mathcal D$ is a strictly positive dissipation term that has an explicit
expression in terms of the solution of the diffusive equation that
govern macroscopically the transformation (cf. \eqref{eq:32}). In the
quasi-static limit we prove that $\mathcal D \to 0$.
A similar interpretation of quasi-static
transformations, for thermodynamic systems with one parameter (density),
has been proposed in recent works by Bertini et al. \cite{bertini-prl,
  bertini-jsp}. 

In the case of the harmonic chain, the thermodynamic entropy is a
function of the temperature, so it remains constant in isothermal
transformation. Then heat is equal to the dissipation term $\mathcal
D$. It means that in the quasistatic limit for the harmonic chain,
there is no heat produced, internal energy is changed by work in a
perfectly efficient way.  

Thermodynamics does not specify the time scale for the
transformations, this may depend on the nature of the transformation
(isothermal, adiabatic, ...) and the details of the microscopic system
and of the exterior agent (heat bath etc.). In this system of
oscillators, in adiabatic setting, with also momentum conservation,
the relevant space--time scale is hyperbolic (cf. \cite{oe}). 

The proof of the hydrodynamic limit follows the lines of \cite{ov,
  trem}, using the relative entropy method (cf. \cite{yau,
  kipnis}). The method has to be properly adapted to deal with the
boundary conditions.


\section{Isothermal microscopic dynamics}
\label{sec:isoth-micr-dynam}

We consider a chain of $N$ coupled oscillators in one dimension.
Each particle has the same mass that we set equal to 1.
The position of atom $i$ is denoted by $q_i\in \mathbb R$, while its
momentum is denoted by $p_i\in\mathbb R$. 
 Thus the configuration space is $(\mathbb R\times \mathbb R)^N$. 
We assume that an extra particle $0$ to be attached to a fixed point
and does not move, i.e. $(q_0,p_0)\equiv(0,0)$,
 while on particle $N$ we apply a force $\tilde\tau(t)$ depending on time.
Observe that only the particle 0 is constrained to not move, and that
$q_i$ can assume also negative values. 

Denote by  ${\bf q} :=(q_1,\dots,q_N)$ and ${\bf p}
:=(p_1,\dots,p_N)$. The interaction between two particles $i$ and
$i-1$ will be described by the potential energy $V(q_i-q_{i-1})$ of an
anharmonic spring relying the particles. We assume   $V$ to be a
positive smooth function which for large $r$ grows faster than linear
but at most quadratic, that means that there exists a constant $C>0$
such that 
\begin{eqnarray}\label{V1}
&&\lim_{|r|\rightarrow\infty}\frac{V(r)}{|r|}=\infty.
\\\nonumber\\
&&\label{V2}
\limsup_{|r|\rightarrow\infty}V^{\prime\prime}(r)\leq C<\infty.\\\nonumber\\\nonumber
\end{eqnarray}
Energy is defined by the following Hamiltonian:
\begin{eqnarray*}
\mathcal H_N(\bf q ,\bf p ):
&=&\sum_{i=1}^{N}  \left( \frac{p_i^2}2 + V(q_{i}-q_{i-1}) \right).
\end{eqnarray*}
Since we focus on a nearest neighbor interaction, we may define the distance between particles by \index{{\large{CHAPTER 2:}}!$r_i$}
$$
r_i=q_{i}-q_{i-1}, \qquad i=1,\dots,N.
$$ 
The chain is immersed in a thermal bath at temperature $\beta^{-1}$
that we model by the action of $N$ independent Langevin processes.
The dynamics is defined by the solution of the system of stochastic
differential equations
\begin{equation}
  \label{eq:sde}
  \begin{split}
    dr_i &= N^2 (p_i - p_{i-1})\; dt\\
    dp_i &= N^2 (V'(r_{i+1}) - V'(r_i))\; dt - N^2 \gamma p_i \; dt -
    N \sqrt{2\gamma\beta^{-1}} dw_i, \qquad i=1,\dots, N-1,\\
    dp_N &= N^2(\tilde\tau(t) - V'(r_N))\; dt - N^2\gamma p_N \; dt -
    N \sqrt{2\gamma\beta^{-1}} dw_N
  \end{split}
\end{equation}
Here $\{w_i(t)\}_i$ are N-independent standard Wiener processes,
$\gamma>0$ is a parameter of intensity of the interaction with the heat
bath, $p_0$ is set identically to $0$. We have also already rescaled
time according to the diffusive space-time scaling. Notice that
$\tilde\tau(t)$ changes at this macroscopic time scale.

The generator of this diffusion is given by
\begin{equation}\label{noisygen}
\mathcal L_N^{\tilde\tau(t)}:= N^2 A^{\tau(t)}_N + N^2 \gamma S_N.
\end{equation}
Here the Liouville operator  $A^{\tau}_N$ is given by
\begin{eqnarray}\label{liouville}
\nonumber A^{\tau}_N
 &=&\nonumber\sum_{i=1}^{N}(p_{i}-p_{i-1})\frac{\partial}{\partial
   r_i}+\sum_{i=1}^{N-1}\left(V^{\prime}(r_{i+1})-V^{\prime}(r_{i})\right)
 \frac{\partial}{\partial p_i}\\ 
 &&+\left(\tau -V^{\prime}(r_N)\right)\frac{\partial}{\partial p_N},
 \end{eqnarray}
while 
\begin{equation}
  \label{eq:3}
  S = \sum_{i=1}^N \left(\beta^{-1} \partial_{p_i}^2 -
    p_i \partial_{p_i} \right) 
\end{equation}

For $\tilde\tau(t) = \tau$ constant, the system has a unique
stationary measure given by the product
\begin{equation}
  \label{eq:gibbs}
  d\mu^N_{\tau,\beta} = \prod_{i=1}^N e^{-\beta (\mc E_i - \tau
      r_i) - \mc G_{\tau,\beta}}\; dr_i\; dp_i = g^N_\tau  d\mu^N_{0,\beta}
\end{equation}
where we denoted $\mc E_i = p_i^2/2 + V(r_i)$, the energy we attribute
to the particle $i$, and 
\begin{equation}
  \label{eq:pfunct}
  \mc G_{\tau,\beta} = \log \left[\sqrt{2\pi\beta^{-1}}\int e^{-\beta
      (V(r) - \tau r)}\; dr  \right].
\end{equation}
Observe that the function $\mf r(\tau) = \beta^{-1} \partial_\tau \mc
G_{\tau,\beta}$ gives the average equilibrium length in function of
the tension $\tau$, and we denote the inverse by $\bs \tau(\mf r)$.

We will need also to consider local Gibbs measure (inhomogeneous
product), corresponding to profiles of tension $\{\tau(x),
x\in[0,1]\}$:
\begin{equation}
  \label{eq:gibbs}
  d\mu^N_{\tau,\beta} = \prod_{i=1}^N e^{-\beta (\mc E_i - \tau(i/N)
      r_i) - \mc G_{\tau(i/N),\beta}}\; dr_i\; dp_i 
    = g^N_{\tau(\cdot)} d\mu^N_{0,\beta}
\end{equation}

Given an initial profile of tension $\tau(0,x)$,
we assume that initial probability state is given by an absolutely
continuous measure (with respect to the Lebesgue measure),
 whose density with respect to $d\mu^N_{0,\beta}$ is given by
 $f^N_{0}$, such that the relative entropy 
 with respect to $\mu^N_{\tau(0,x),\beta}$
 \begin{equation}
   \label{eq:1}
   H_N(0) = \int f_{0}^N \log
   \left(\frac{f_{0}^N}{g^N_{\tau(0,\cdot)}}\right)  d\mu^N_{0,\beta}
 \end{equation}
satisfies
\begin{equation}
  \label{eq:2}
  \lim_{N\to\infty} \frac {H_N(0)}N = 0
\end{equation}
This implies the following convergence in probability with respect to $f_0^N$: 
\begin{equation}
  \label{eq:4}
  \frac 1N \sum_{i=1}^N G(i/N) r_i(0) \longrightarrow \int_0^1 G(x)
  \mf r(\tau(0,x))\; dx
\end{equation}

The macroscopic evolution for the stress will be given by
\begin{equation}
  \label{eq:diff}
  \begin{split}
    \partial_t r(t,x) &= \gamma^{-1}\partial_x^2 \bs\tau(r(x,t)), \quad
    x\in[0,1]\\
    \partial_x r (t,0) &= 0, \quad \bs\tau(r(t,1)) = \tilde\tau(t),
    \quad t>0\\
    \bs\tau(r(0,x)) &= \tau(0,x), \quad x\in[0,1]
  \end{split}
\end{equation}
Observe that we do not require that $\tau(r(0,1)) = \tilde\tau(0)$, so
we can consider initial profiles of equilibrium with tension different
than the applied $\tilde\tau$. 

The main result is the following
\begin{theorem}\label{main}
  \begin{equation}
    \label{eq:7}
    \lim_{N\to\infty} \frac {H_N(t)}N = 0
  \end{equation}
where
\begin{equation}
  \label{eq:5}
    H_N(t) = \int f_{t}^N \log
   \left(\frac{f_{t}^N}{g^N_{\tau(t,\cdot)}}\right)\; d\mu^N_{0,\beta}
\end{equation}
with $\tau(t,x) = \bs \tau(r(t,x))$, and $f^N_t$ the density of the
configuration of the system at time $t$.
\end{theorem}

A schetch of the proof is postponed to section
\ref{sec:proof-hydr-limit}. 

\begin{remark}
  The proof and the result are identical (up to some constant) if we
  use a different modelling of 
  the heat bath, where the particles undergo independent random
  collisions such that 
  after the collision they get a new value distributed by a gaussian
  distribution with variance $\beta^{-1}$, i.e.
  \begin{equation}
    \label{eq:35}
    Sf(\br,\bp) =  \sum_{i=1}^N \int \left(f(\br,p_1, \dots, p_i',\dots) -
    f(\br, \bp)\right) \frac{e^{-\beta (p_i')^2/2}}{\sqrt{2\pi\beta^{-1}}} dp_i'
  \end{equation}
\end{remark}

\section{Thermodynamic consequences}
\label{sec:therm-cons}

Consider the case where we start our system with a constant tension 
$\tau(0,x) = \tau_0$ and we apply a tension $\tilde\tau(t)$ going
smoothly from $\tilde\tau(0) = \tau_0$ to $\tilde\tau(t) = \tau_1$ for
$t\ge t_1$. 
It follows from standard arguments that
\begin{equation}
  \label{eq:6}
  \lim_{t\to\infty} \bs\tau(r(t,x)) \ =\ \tau_1, \qquad \forall x\in[0,1]
\end{equation}
so on an opportune time scale, this evolution represents an isothermal
thermodynamic transformation from the equilibrium state
$(\tau_0,\beta^{-1})$ to $(\tau_1,\beta^{-1})$. Clearly this is an
irreversible transformation and will statisfy a strict Clausius
inequality.   

The length of the system at time $t$ is given by
\begin{equation}
  \label{eq:8}
  L(t) = \int_0^1 r(t,x) \; dx
\end{equation}
and the work done by the force $\tilde\tau$:
\begin{equation}
  \label{eq:9}
  \begin{split}
    W(t) = \int_0^t \tilde\tau(s) dL(s) = \gamma^{-1} \int_0^t ds\;
    \tilde\tau(s)\; \int_0^1 dx\; \partial_x^2 \bs\tau(r(s,x)) \\
    = \gamma^{-1} \int_0^t \tilde\tau(s) \partial_x \bs\tau(r(s,1)) ds
  \end{split}
\end{equation}
The free energy of the equilibrium state $(r,\beta)$ is given by the
Legendre transform of $\beta^{-1}\mc G_{\tau,\beta}$:
\begin{equation}
  \label{eq:freeen}
   F(r,\beta) = \inf_\tau \left\{ \tau r - \beta^{-1}\mc
    G_{\tau,\beta}\right\} 
\end{equation}
Since $\beta$ is constant, we will drop the dependences on it in the
following. It follows that $\bs\tau(r) = \partial_r \mc F$.
Thanks to the local equilibrium, we can define the free energy at time
$t$ as
\begin{equation}
  \label{eq:11}
  \mc F(t) = \int_0^1 F(r(t,x),\beta) \; dx .
\end{equation}
Its time derivative is (after integration by parts):
\begin{equation*}
  \frac d{dt}  \mathcal F (t) = -\gamma^{-1}\int_0^1 \left(\partial_x
    \bs\tau(r(t,x))\right)^2 \; dx + \gamma^{-1}\tilde\tau(t) \partial_x
  \bs\tau(r(t,x))\big|_{x=1} 
\end{equation*}
i.e.
\begin{equation*}
   \mathcal F (t) - \mathcal F (0) =
     W(t) -
    \gamma^{-1} \int_0^t ds \int_0^1 \left(\partial_{x} \bs\tau(r(s,x))\right)^2
    \; dx  
\end{equation*}
Because or initial condition, $\mathcal F (0) = F(\tau_0)$, and
because \eqref{eq:6} we have $\mathcal F (t) \to F(\tau_1)$, and we
conclude that
\begin{equation}
  \label{eq:10}
   F(\tau_1) -  F(\tau_0) = W -  \gamma^{-1}\int_0^{+\infty} ds \int_0^1
   \left(\partial_{x} \bs\tau(r(s,x))\right)^2 \; dx  
\end{equation}
where $W$ is the total work done by the force $\tilde\tau$ in the
transformation up to reaching the new equilibrium and is expressed by
taking the limit in \eqref{eq:9} for $t\to\infty$:
\begin{equation}
  \label{eq:15}
   W = \int_0^\infty \tilde\tau(s) dL(s) =  \gamma^{-1}\int_0^\infty
  \tilde\tau(s)  \partial_x \bs\tau(r(s,1)) ds  
\end{equation}
By the same argument we will use in the proof of Proposition \ref{quasistatic}
we have that the second term of the righthand side of \eqref{eq:10} is
finite, that implies the existence of $W$. 

 Since the second
term on right hand side is always strictly positive, we have obtained
a strict Clausius inequality. This is not surprizing since we are
operating an irreversible transformation.  

If we want to obtain a \emph{reversible quasistatic isothermal
transformation}, we have introduce another
larger time scale, i.e. introduce a small parameter $\ve >0$ and
apply a tension slowly varying in time $\tilde \tau(\ve t)$
.
The diffusive equation becomes
\begin{equation}
  \label{eq:41slow}
  \partial_{t} r_\ve(t,x) = \gamma^{-1}\partial_x^2 \bs\tau(r_\ve(t,x))
\end{equation}
 with boundary conditions
\begin{equation}
  \label{eq:48slow}
  \begin{split}
    &\partial_x r_\ve(t,0) = 0\\
    &\bs\tau (r_\ve(t,1)) = \tilde\tau(\ve t)
  \end{split}
\end{equation}

Then \eqref{eq:10} became
\begin{equation}
  \label{eq:43}
  F(r_1) - F(r_0) = W_\ve - \gamma^{-1}\int_0^\infty ds 
  \int_0^1 \left(\partial_{x} \bs\tau(r_\ve(s,x))\right)^2 \; dx  
\end{equation}
\begin{proposition}\label{quasistatic}
  \begin{equation}
    \label{eq:12}
    \lim_{\ve\to 0} \int_0^\infty ds 
  \int_0^1 \left(\partial_{x} \bs\tau(r_\ve(s,x))\right)^2 \; dx  = 0
  \end{equation}
\end{proposition}

\begin{proof}
To simplify notations, let set here $\gamma=1$. 
We look at the time scale $\ft = \ve^{-1} t$, then 
$\tilde r_\ve(\ft,x) = r_\ve(\ve^{-1} t,x)$ statisfy the equation
\begin{equation}
  \label{eq:41fast}
  \partial_{\ft} \tilde r_\ve(\ft,x) = \ve^{-1}\partial_x^2
  \bs\tau(\tilde r_\ve(\ft,x))
\end{equation}
 with boundary conditions
\begin{equation}
  \label{eq:48fast}
  \begin{split}
    &\partial_x r_\ve(\ft,0) = 0\\
    &\bs\tau (r_\ve(\ft,1)) = \tilde\tau(\ft)
  \end{split}
\end{equation}

  \begin{equation}\label{eq:lyap}
    \begin{split}
      \frac12\int_0^1 &\left(\tilde r_\ve(\ft,x) - \mf r[\tilde
        \tau(\ft)]\right)^2   d\ft  \\
      =& \int_0^\ft ds \int_0^1 dx \;
      \left(\tilde r_\ve(s,x) - \mf r(\tilde \tau(s))\right)
      \left(\ve^{-1}\partial_x^2 \bs\tau[\tilde r_\ve(s,x)] - 
        \frac d{ds} \mf r[\tilde \tau(s)] \right)\\
      =& - \ve^{-1} \int_0^\ft ds \int_0^1 dx \;\left(\partial_x
        \tilde r_\ve(s,x)\right)^2 \frac {d\bs\tau}{dr}
      \left[\tilde r_\ve(s,x)\right]  \\
      &- \int_0^t ds \frac {d\mf r}{d\tau}(\tilde \tau(s))
      \tilde\tau'(s) \int_0^1 dx\; \left(\tilde r_\ve(s,x) - 
          \tilde r_\ve(s,1)\right) 
    \end{split}
  \end{equation}

Rewriting
\begin{equation*}
  \begin{split}
    \left|\int_0^1 dx\; \left(\tilde r_\ve(s,x) -\tilde r_\ve(s,1)\right)\right| =
    \left|\int_0^1 dx\; \int_x^1 dy\; \partial_y \tilde r_\ve(s,y)\right| \\
    = \left|\int_0^1 dy\; y \partial_y \tilde r_\ve(s,y) \right| 
    \le \frac{\alpha}{2\ve} \int_0^1 dx \;\left(\partial_x
       \tilde r_\ve(s,x)\right)^2  + \frac \ve{4\alpha}
  \end{split}
\end{equation*}
Recall that the free energy $F$ is strictly convex and that 
$0< C_- \le\frac {d\mf r}{d\tau} \le C_+ < + \infty$, and furthermore
we have chosen $\tilde\tau$ such that $\left|\tilde\tau'(t)\right|\le
1_{t\le t_1}$. 
Regrouping positive terms on the left hand side we obtain the bound:
\begin{equation}
  \label{eq:13}
  \begin{split}
    \frac12\int_0^1 \left(\tilde r_\ve(\ft,x) - \mf r[\tilde
      \tau(\ft)]\right)^2  dx + \ve^{-1} \left(C_- - \frac{C_+ \alpha \ft}2\right)
     \int_0^\ft ds \int_0^1 dx \;\left(\partial_x \tilde r_\ve(s,x)\right)^2  
      \le \frac {\ve C_+ \ft}{4\alpha} 
  \end{split}
\end{equation}
By choosing $\alpha = \frac{C_-}{C_+ \ft} $, we obtain, for any $\ft> t_1$:
\begin{equation}
  \label{eq:16}
 \frac1{C_-}\int_0^1 \left(\tilde r_\ve(\ft,x) - \mf r[\tilde
      \tau(t_1)]\right)^2  dx + \ve^{-1} 
     \int_0^\ft ds \int_0^1 dx \;\left(\partial_x \tilde r_\ve(s,x)\right)^2  
      \le \frac {\ve}{2} 
\end{equation}
then we can take the limit as $\ft\to\infty$, the first term on the
right hand side of \eqref{eq:16} will disappear, and we obtain
\begin{equation}
  \label{eq:17}
  \ve^{-1} \int_0^{+\infty} ds \int_0^1 dx \;\left(\partial_x \tilde
    r_\ve(s,x)\right)^2 \le \frac {\ve}{2} 
\end{equation}
that implies \eqref{eq:12}.
\end{proof}
Consequently we obtain the Clausius identity for the quasistatic
reversible isothermal transformation.  

Along the lines of the proof above it is also easy to prove that
\begin{equation}
  \label{eq:14}
  \lim_{\ve\to 0} \int_0^1 \left(r_\ve(t,x) - \mf r[\tilde \tau(\ve
        t)]\right)^2   dx = 0
\end{equation}
that gives a rigorous meaning to the \emph{quasistatic} definition.

The internal energy of the thermodynamic equilibrium state $(r,T)$ is
defined as $U = F + TS$, where $S$ is the thermodynamic entropy.
The first principle of thermodynamics defines the heat $Q$ transferred
as $\Delta U = W + Q$.

The change of internal energy in the isothermal tranformation is
given by
\begin{equation}
  \label{eq:18}
  \Delta U = \Delta F + T \Delta S = W - \gamma^{-1} \int_0^{+\infty}
  ds \int_0^1 
  dx \;\left(\partial_x \bs\tau(r(s,x))\right)^2 + T \Delta S
\end{equation}
Then for the irreversible transformation we have 
$Q\le T\Delta S$, while equality holds in the quasistatic limit. 

The linear case is special, it corresponds to the microscopic
harmonic interaction. In this case $S$ is just a function of the
temperature ($S \sim \log T$), so $\Delta S = 0$ for any isothermal
transformation. Correspondingly the heat exchanged with the thermostat
is always negative and given by $Q =- \gamma^{-1} \int_0^{+\infty} ds \int_0^1
  dx \;\left(\partial_x r(s,x)\right)^2$, and null in the quasistatic
  limit.  

\section{Work and Microscopic Heat}
\label{sec:microscopic-heat}

The microscopic total lenght is defined by $q_N = \sum_i r_i$, the
position of the last particle. To connect it to the macroscopic space
scale we have to divide it by $N$, so se define
\begin{equation}
  \label{eq:22}
  \mc L_N(t) = \frac {q_N(t)}N = \frac 1N \sum_{i=1}^N r_i(t) .
\end{equation}
 The time evolution in the scale considered is given by
 \begin{equation}
   \label{eq:23}
   \mc L_N(t) - \mc L_N(0) = \int_0^t Np_N(s) \; ds. 
 \end{equation}
If we start with the equilibrium distribution with length $r_0$, the
law of large numbers guarantees that
\begin{equation}
  \label{eq:24}
  \mc L_N(0) \mathop{\longrightarrow}_{N\to\infty} r_0,
\end{equation}
in probability.

By theorem  \ref{main}, we also have the convergence at time t:
\begin{equation}
  \label{eq:20}
   \mc L_N(t) \mathop{\longrightarrow}_{N\to\infty} L(t)
   \mathop{\longrightarrow}_{t\to\infty} r_1 = \frak r(\tau_1), 
\end{equation}
where $ L(t) $ is defined by \eqref{eq:8}. 
Notice that in \eqref{eq:23} while $Np_N(s)$ fluctuates wildly as
$N\to\infty$,  its
time integral is perfectly convergent and in fact converges to a
deterministic quantity.

The microscopic work done up to time $t$ by the force $\tilde\tau$ is given by
\begin{equation}
  \label{eq:25}
  \mc W_N(t) = \int_0^t \tilde\tau(s) d\mc L_N(s) 
  = \int_0^t \tilde\tau(s) N p_N(s) ds  
\end{equation}
We adopt here the convention that positive work means energy increases
in the system. Notice that $\mc W_N(t)$ defines the actual microscopic
work divided by $N$. 

It is a standard exercise to show that, since $\tilde\tau(t)$ and
$L(t)$ are smooth functions of $t$, by \eqref{eq:20} it follows that 
\begin{equation}
  \label{eq:26}
  \mc W_N(t) \mathop{\longrightarrow}_{N\to\infty} W(t) = \int_0^t
  \tilde\tau(s) dL(s)  
\end{equation}
given by $\eqref{eq:9}$.

Microscopically the energy of the system is defined by 
\begin{equation}
  \label{eq:28}
  E_N  = \frac 1N\sum_i \mc E_i
\end{equation}
Energy evolves in time as
\begin{equation}
  \label{eq:21}
  \begin{split}
    E_N(t) - E_N(0) &= \mc W_N(t) + \mc Q_N(t)\\
   \mc Q_N(t) = - \gamma \int_0^t N\sum_{i=1}^N &
    \left(p_i^2(s) - T\right) \; ds  
   + \sqrt{2\gamma \beta^{-1}} \sum_{i=1}^N \int_0^t p_i(s) dw_i(s)
  \end{split}
\end{equation}
where $\mc Q_N$ is the energy exchanged
with the heat bath, what we call \emph{heat}. 

The law of large numbers for the initial distribution gives
\begin{equation*}
  E_N (0)\ \mathop{\longrightarrow}_{N\to\infty}\
    U(\beta, \tau_0) 
\end{equation*}
in probability. 
By the hydrodynamic limit, we expect that
\begin{equation}
  \label{eq:27}
  \begin{split}
    E_N (t)\ \mathop{\longrightarrow}_{N\to\infty}\
    \int_0^1 U(\beta, \bs\tau(r(t,x)))\; dx
    \ \mathop{\longrightarrow}_{t\to\infty}\  U(\beta, \tau_1).
  \end{split}
\end{equation} 
This is not a consequence of
Theorem \ref{main}, because the relative entropy does not control the
convergence of the energy. In the harmonic case it can be proven by
using similar argument as in \cite{simon} (in fact in this case
$f_N(t)$ is a gaussian distribution where we have control of any moments).

Assuming \eqref{eq:27}, we have that $Q_N(t)$ converges, as
$N\to\infty$, to the deterministic 
\begin{equation}
  \label{eq:29}
  Q(t) =  \int_0^1 \left[U(\beta, \bs\tau(r(t,x))) -  U(\beta,
    \tau_0)\right]\;  dx -  W(t) 
\end{equation}
and as $t\to\infty$:
\begin{equation}
  \label{eq:30}
  Q = U(\beta, \tau_1) - U(\beta, \tau_0) - W, \qquad 
  \text{(first principle)}.
\end{equation}

Recalling that the free energy is equal to $F = U - \beta^{-1} S$,
then we can compute the variation of the entropy $S$ as
\begin{equation}
  \label{eq:31}
  \beta^{-1}(S_1 - S_0 ) = -(F_1 - F_0) + W + Q 
\end{equation}
or also that
\begin{equation}
  \label{eq:32}
  Q = \beta^{-1}(S_1 - S_0 ) - \gamma^{-1} \int_0^\infty dt \int_0^1
  dx \left(\partial_x\bs\tau(r(t,x))\right)^2
\end{equation}

In the quasi static limit, we have seen that $ F_1 - F_0 = W$, and
consequently $\beta Q = S_1 - S_0$, in accord to what thermodynamics
prescribe for quasistatic transformations. 

\begin{remark}
  Assume that the distribution of $p_i(t)$ is best approximated by
$$
e^{\frac{\beta}{N\gamma}\sum_i \partial_x\tau(t, i/N) p_i} g^N_{\tau(t,\cdot)}
\prod_{i=1}^N \; dr_i\; dp_i
$$ 
properly normalized. Then the average of $p_i$ is
 $\frac 1{N\gamma} \partial_x\tau(t, i/N)$, 
and \eqref{eq:21} can be rewritten as
\begin{equation*}
  \begin{split}
    N\gamma \sum_{i=1}^N \left(\left(p_i(t)- \frac
        1{N\gamma} \partial_x\tau(t, i/N)\right)^2 - \beta^{-1}\right)&\\
    - \frac 1{N\gamma} \sum_{i=1}^N \partial_x\tau(t, i/N)^2 &+
    2 \sum_{i=1}^N \partial_x\tau(t, i/N) p_i(t)
  \end{split}
\end{equation*}
Taking expectation, the first term is null (as well as the martingale
not written here) while the last two terms converge to
$\gamma^{-1}\int_0^1 (\partial_x\tau(t,x))^2 dx$. This is correct only
in the harmonic case, i.e. the fluctuation inside
the time integral are very important in order to get the changes in
entropy $S$.
\end{remark}

\section{Proof of the hydrodynamic limit}
\label{sec:proof-hydr-limit}

Define the modified local Gibbs density
\begin{equation}
  \label{eq:19}
 \tilde g^N_{\tau(t,\cdot)} = e^{\frac{\beta}{\gamma N}
   \sum_i \partial_x\tau(t, i/N) p_i} g^N_{\tau(t,\cdot)} Z_{N,t}^{-1} 
\end{equation}
where $Z_{N,t}$ is a normalization factor. Then define the
corresponding relative entropy 
\begin{equation}
  \label{eq:33}
  \tilde H_N(t) = \int f_{t}^N \log
   \left(\frac{f_{t}^N}{\tilde g^N_{\tau(t,\cdot)}}\right)\; d\mu^N_{0,\beta}
\end{equation}
It is easy to see that $\lim_{N\to\infty} N^{-1} \left(\tilde H_N(t)-
  H_N(t)\right) = 0$.

Computing the time derivative
\begin{equation}
  \label{eq:34}
\frac d{dt}  \tilde H_N(t) = \int f_{t}^N \mathcal
L_N^{\tilde\tau(t)} \log f_{t}^N \; d\mu^N_{0,\beta} - \int   f_{t}^N
\left( \mathcal L_N^{\tilde\tau(t)} + \partial_t\right) 
\log \tilde g^N_{\tau(t,\cdot)} \; d\mu^N_{0,\beta}
\end{equation}
Using the inequality
\begin{equation*}
  f_{t}^N \mathcal L_N^{\tilde\tau(t)} \log f_{t}^N \le 
\mathcal L_N^{\tilde\tau(t)} f_{t}^N
\end{equation*}
and since $d\mu^N_{0,\beta}$ is stationary for $\mathcal L_N^{0}$, we have
\begin{equation*}
  \int f_{t}^N \mathcal
L_N^{\tilde\tau(t)} \log f_{t}^N \; d\mu^N_{0,\beta}  \le N^2 \tau
\int \partial_{p_N} f_{t}^N \; d\mu^N_{0,\beta} =
N^2 \tau \beta\int p_N f_{t}^N \; d\mu^N_{0,\beta} 
\end{equation*}

By explicit calculation
\begin{equation*}
  \begin{split}
    \mathcal L_N^{\tilde\tau(t)} \log \tilde g^N_{\tau(t,\cdot)} =
    -\beta N^2 \sum_i \tau(i/N, t) (p_i - p_{i-1})
    \\
    + \beta\gamma^{-1} N \sum_i \partial_x\tau(t, i/N)
    \left(V'(r_{i+1}) - V'(r_i) \right) - N \beta \sum_i \partial_x\tau(t, i/N)
    p_i \\
    = N^2 \tau \beta p_N  - \beta\gamma^{-1}
    \sum_i \partial^2_x\tau(t, i/N) V'(r_i) + o(N) 
  \end{split}
\end{equation*}
and
\begin{equation*}
  \partial_t \log \tilde g^N_{\tau(t,\cdot)} = -\beta\sum_i \partial_t
  \tau(t, i/N) (r_i - r(t,i/N)) + O(1). 
\end{equation*}
Then we can estimate
\begin{equation*}
  \begin{split}
    \frac d{dt} \tilde H_N(t) \le \beta \int \sum_i
    \left[\gamma^{-1} \partial^2_x\tau(t, i/N) V'(r_i) + \partial_t
      \tau(t, i/N) (r_i - r(t,i/N))\right] f^N_t d\mu^N_{0,\beta}\\ +
    o(N)
  \end{split}
\end{equation*}
and the rest of the proof follows by the standard arguments of the
relative entropy method (cf. \cite{oe, trem, kipnis, yau}).

\bibliographystyle{amsalpha}

\providecommand{\bysame}{\leavevmode\hbox to3em{\hrulefill}\thinspace}
\providecommand{\MR}{\relax\ifhmode\unskip\space\fi MR }
\providecommand{\MRhref}[2]{%
  \href{http://www.ams.org/mathscinet-getitem?mr=#1}{#2}
}
\providecommand{\href}[2]{#2}
\begin{thebibliography}{}

\end{thebibliography}


\begin{thebibliography}{A}
\bibitem{bertini-prl} L Bertini, D Gabrielli, G Jona-Lasinio, C
  Landim, Clausius inequality and optimality of quasistatic
  transformations for nonequilibrium stationary states,
  Phys. Rev. Lett. (2013).
\bibitem{bertini-jsp}  L Bertini, D Gabrielli, G Jona-Lasinio, C
  Landim, Thermodynamic transformations of nonequilibrium states,
  J. Stat. Phys. (2012).
\bibitem{kipnis} Kipnis, C. and Landim, C.,
{\it Scaling Limits of Interacting Particle Systems},
Springer-Verlag: Berlin, 1999.

\bibitem{oe} N. Even, S. Olla, Hydrodynamic Limit for an Hamiltonian
  System with Boundary Conditions  and Conservative Noise,
  http://arxiv.org/abs/1009.2175v1. 

\bibitem{ov} S. Olla, S. Varadhan, Hydrodynamical limit for
  Ornstein-Uhlenbeck interacting Particles, Commun. Math. Phys. {\bf
    135}, 355–-378 (1991).
\bibitem{simon} M. Simon, Hydrodynamic limit for the velocity-flip
  model, Stochastic Processes and their Applications {\bf 123} (2013)
  3623 – 3662. 

\bibitem{trem} C. Tremoulet, Hydrodynamic limit for interacting
  Ornstein–Uhlenbeck particles, Stoch.Proc. and App., \textbf{102}, vol
  1, 139–-158 (2002).
\bibitem{yau}  Yau, H.~T., Relative entropy and hydrodynamics of
  Ginzburg-Landau models, Lett. Math. Phys. {\bf 22(1)} (1991), 63--80.
\end{thebibliography}

\end{document}